\documentclass[10pt,conference]{IEEEtran}

\usepackage[utf8]{inputenc} 
\usepackage[T1]{fontenc}
\usepackage{url}
\usepackage{ifthen}
\usepackage{cite}
\usepackage[cmex10]{amsmath} % Use the [cmex10] option to ensure complicance
\IEEEoverridecommandlockouts
\usepackage{makecell}
\usepackage{cite}
\usepackage{amsmath,amssymb,amsfonts}

\usepackage{algorithmic}
\usepackage{graphicx}
\usepackage{textcomp}
\usepackage{xcolor}
\newtheorem{theorem}{Theorem}

\usepackage{bbm}
\usepackage{bm}
\newtheorem{proof}{Proof}
\newtheorem{corollary}{Corollary}
\usepackage{bbm} 

\usepackage[letterpaper, top=0.75in, bottom=1.1in, left=0.67in, right=0.7in]{geometry}
\begin{document}
\setlength{\textfloatsep}{0.1cm}
\setlength{\abovedisplayskip}{0.2cm}
\setlength{\belowdisplayskip}{0.15cm}
\setlength{\baselineskip}{0.42cm}
\setlength{\abovecaptionskip}{1mm}

\title{PAPR Analysis for MIMO FTN Signaling with Gaussian Symbols} 

%\IEEEspecialpapernotice{\textit{(Invited Paper)}}

% %%% Single author, or several authors with same affiliation:
% \author{%
%  \IEEEauthorblockN{Author 1 and Author 2}
% \IEEEauthorblockA{Department of Statistics and Data Science\\
%                    University 1\\
 %                   City 1\\
  %                  Email: author1@university1.edu}% }
\author{Zichao~Zhang\textsuperscript{\P},
		~Melda~Yuksel\textsuperscript{\S},~Gokhan~M.~Guvensen\textsuperscript{\S},~Halim~Yanikomeroglu\textsuperscript{\P}  \\
  \textsuperscript{\P}Department of  Systems and Computer Engineering, Carleton University, Ottawa, ON, Canada  \\
  \textsuperscript{\S}Department of Electrical and Electronics Engineering, Middle East Technical University, Ankara, Turkey   \\ \\
 Emails: zichaozhang@cmail.carleton.ca, \{ymelda, guvensen\}@metu.edu.tr, halim@sce.carleton.ca 
}

%%% Several authors with up to three affiliations:

\maketitle

%%%%%%
%% Abstract: 
%% If your paper is eligible for the student paper award, please add
%% the comment "THIS PAPER IS ELIGIBLE FOR THE STUDENT PAPER
%% AWARD." as a first line in the abstract. 
%% For the final version of the accepted paper, please do not forget
%% to remove this comment!
%%

\begin{abstract}
    Faster-than-Nyquist signaling serves as a promising solution for improving spectral efficiency in future generations of communications. However, its nature of fast acceleration brings highly overlapped pulses that lead to worse peak-to-average power ratio (PAPR) performance. In this paper, we investigate the PAPR behavior of MIMO FTN using Gaussian symbols under optimal power allocation for two power constraints: fixed transmit power and fixed received signal-to-noise-ratio (SNR). Our findings reveal that %for fixed transmit SNR, the PAPR behavior stays the same regardless of the acceleration. However, if the received SNR is fixed, the PAPR distribution becomes unbounded.
    PAPR is mainly determined by the acceleration factor and the power constraint, but power allocation optimization does not change the PAPR behavior for Gaussian signaling. 
\end{abstract}

\section{Introduction}

In the next generation of wireless communication systems, an unprecedented array of new use cases and applications are emerging, including holographic communications, tactile and haptic internet applications, and space-terrestrial integrated networks \cite{6g}. These advancements demand significantly higher bandwidth, presenting a substantial challenge for communication systems. To meet the stringent bandwidth requirements of 6G and beyond, transmissions will also happen in higher-frequency bands, such as the terahertz (THz) spectrum \cite{6gband}. However, bandwidth as a resource is limited, innovative strategies are necessary to enhance spectral efficiency. Faster-than-Nyquist (FTN) signaling has gained prominence as a promising solution to alleviate bandwidth scarcity.

Originally proposed by Mazo in 1975 \cite{mazo}, FTN signaling has since brought extensive research. FTN signaling increases the transmission rate and offers a notable improvement in spectral efficiency without expanding the bandwidth by speeding up the symbol rate at a factor of $\delta$, which is called the acceleration factor. However, FTN signaling inherently introduces intersymbol interference (ISI) since it breaks the Nyquist no-ISI criterion, which complicates the detection process compared to conventional Nyquist signaling. Despite this challenge, the potential of FTN signaling to revolutionize high-efficiency communication systems continues to drive significant academic and industrial interest.

In FTN signaling, accelerating symbol transmission places pulses closer together, causing significant overlap and a high peak-to-average power ratio (PAPR) \cite{refpapr}. Practical power amplifiers, with finite dynamic ranges \cite{rfref}, cannot handle extreme peaks beyond the saturation point, leading to nonlinear distortion and degraded communication performance. To prevent this, input signal power is reduced to keep peaks within the amplifier’s linear region. This reduction, called power amplifier back-off \cite{backoffref}, is measured in decibels (dB) and is approximately proportional to PAPR in systems without specialized mitigation techniques.  

FTN transmission imposes a different power constraint than conventional Nyquist signaling. For transmission power \( P \), sending \( N \) symbols over \( N\delta T \) seconds results in total energy \( NP\delta T \). The symbol energy \( E = P\delta T \) decreases as \( \delta \) decreases, lowering the average energy per symbol if power remains constant \cite{zhang2024capacitypapranalysismimo}. This reduction in symbol energy decreases Euclidean distances between constellation points, increasing noise susceptibility and symbol error probability.  

This paper examines PAPR in FTN transmission under optimal power allocation for two power constraints: fixed transmit power and fixed receive signal-to-noise ratio (SNR). While \cite{zhang2024capacitypapranalysismimo} analyzed PAPR with uniform power allocation only, this study extends the analysis to optimal power allocation. Additionally, although \cite{zhang2024capacitypapranalysismimo} derived the optimal input power spectrum (i.e., in frequency) for $\delta < \frac{1}{1+\beta}$, the corresponding input covariance matrix in the time domain is unknown. This paper determines this optimal input covariance matrix for FTN signaling with $\delta < \frac{1}{1+\beta}$, thus putting the PAPR at the transmitter output.

\section{System Model}

We assume that the transmitter is equipped with \( K \) antennas and the receiver with \( L \) antennas. Each transmit antenna sends \( N=2M+1 \) symbols, denoted as \( a_k[n] \) for \( n=0, \dots, N-1 \) and \( k=1, \dots, K \). All the transmit antennas are equipped with the same pulse-shaping filter \( p(t) \). The symbols will then be passed through the pulse-shaping filters for transmission. 
Since FTN transmission is assumed, symbols are transmitted every \( \delta T \) seconds, where \( T \) is the Nyquist sampling period ensuring no ISI. Thus, the transmitted signal from the \( k \)th antenna is expressed as   
\begin{align}
    x_k(t) &= \sum_{m=-M}^{M} (a_{r,k}[m] + j a_{i,k}[m]) p(t - m\delta T) \label{eqn:defxt}\\ 
    &= x_{r,k}(t) + j x_{i,k}(t),
\end{align}  
where $a_{r,k}[m]$ and $a_{i,k}[m]$ represent the real and imaginary components of \( a_k[m] \) respectively, \( x_{r,k}(t) \) and \( x_{i,k}(t) \) represent the real and imaginary components of \( x_k(t) \).
The transmitted signal propagates through a frequency-flat fading wireless channel. We denote the channel coefficient from the \( k \)th transmit antenna to the \( l \)th receive antenna as \( h_{lk} \in \mathbb{C} \). At the receiver, the received signal at antenna \( l \) consists of the transmitted signals from all the transmit antennas and the additive circularly symmetric complex Gaussian noise \( \xi_l(t) \). This signal is processed by a matched filter \( p^*(-t) \), yielding the output  
\begin{equation}
    y_l(t) = \sum_{k=1}^{K} h_{lk} \sum_{m=0}^{N-1} a_k[m] g(t - m\delta T) + \eta_l(t),
\end{equation}  
where \( g(t) = p(t) \star p^*(-t) \) is the matched filter response, and the filtered noise is \( \eta_l(t) = \xi_l(t) \star p^*(-t) \). Sampling the matched filter output at intervals of \( \delta T \) gives  
\begin{align}
    y_l[n] &= y_l(n\delta T) \notag \\
    &= \sum_{k=1}^{K} h_{lk} \sum_{m=0}^{N-1} a_k[m] g[n - m] + \eta_l[n], \label{eq:revsamp}
\end{align}  
where $g[n-m]=g((n-m)\delta T)$ and $\eta_l[n]=\eta_l(n\delta T)$. At each sample $y[n]$, we can see the interference from other symbols $\sum_{m\neq n}^{N-1} a_k[m] g[n - m]$ since \( g(n-m)\delta T \neq 0 \) for \( n \neq m \).  We can express \eqref{eq:revsamp} in vector form as  
\begin{align}
    \bm{y}_l &= \sum_{k=1}^{K} h_{lk} \bm{G} \bm{a}_k + \bm{\eta}_l \\
    &=\sum_{k=1}^{K} h_{lk} \bm{G} \left(\bm{a}_{r,k}+\bm{a}_{i,k} \right)+ \bm{\eta}_l
\end{align}  
where  
$
\bm{y}_l = [y_l[0],\dots, y_l[N-1]]^T, \quad \bm{a}_k = [a_k[0],\dots, a_k[N-1]]^T,$ $ \quad \bm{\eta}_l = [\eta_l[0],\dots, \eta_l[N-1]]^T
$, $\bm{a}_{r,k} = [a_{r,k}[0],\dots, a_{r,k}[N-1]]^T$, and $\bm{a}_{i,k} = [a_{i,k}[0],\dots, a_{i,k}[N-1]]^T$.  
The \( N \times N \) matrix \( \bm{G} \) is defined by \( (\bm{G})_{n,m} = g[n-m] \), and it is Hermitian. 
We can see that matrix $\bm{G}$ has the structure of a Toeplitz matrix. 
A Toeplitz matrix $\bm{T}$ with size $N\times N$ has entries  $\left(\bm{T}\right)_{i,j}=t_{i-j}, i,j=0,\dots, {N-1}$, which means $\bm{T}$ has the same value on each diagonal. Another important concept we will be using is the generating function of $\bm{T}$. It is defined as
\begin{equation}
    \mathcal{G}(\bm{T})=\sum_{k=-\infty}^{\infty}t_ke^{j2\pi f_nk}, f_n\in \left[-\frac{1}{2},\frac{1}{2}\right].\label{eqn:defgenfunc}
\end{equation}

% By collecting received samples from all antennas, the MIMO FTN channel input-output model is  
% \begin{equation}
%     \bm{Y} = (\bm{H} \otimes \bm{G}) \bm{A} + \bm{\Omega},
% \end{equation}  
% where  
% $
% \bm{Y} = [\bm{y}_1^T, \dots, \bm{y}_L^T]^T, \quad \bm{A} = [\bm{a}_1^T, \dots, \bm{a}_K^T]^T, \quad \bm{\Omega} = [\bm{\eta}_1^T, \dots, \bm{\eta}_L^T]^T.
% $  
% The channel matrix \( \bm{H} \) consists of the channel coefficients \( h_{lk} \) with \( (\bm{H})_{l,k} = h_{lk} \), and we denote \( \Tilde{\bm{H}} = \bm{H} \otimes \bm{G} \) for simplicity.  
The noise vector \( \bm{\eta}_l \) follows a correlated Gaussian distribution,  $\bm{\eta}_l \sim \mathcal{CN}(\bm{0}_N, \sigma_0^2 \bm{G})$,
where \( \bm{0}_N \) is an \( N \times 1 \) zero vector, and \( \sigma_0^2 \) is the power spectral density (PSD) of \( \xi_l(t) \). Unlike Nyquist signaling, where \( \bm{G} \) reduces to the identity matrix and noise terms remain independent, FTN introduces noise correlation.

% Since \( \xi_l(t) \) are independent across antennas, the matched filter output noise terms \( \eta_l(t) \) are also independent for different \( l \). Thus, the overall noise vector \( \bm{\Omega} \) has the distribution  
% \[
% \bm{\Omega} \sim \mathcal{CN}(\bm{0}_{LN}, \sigma_0^2 (\bm{I}_L \otimes \bm{G})).
% \]  

\section{Instant Power to Average Ratio Analysis}

To analyze the PAPR behavior of FTN transmission, we consider the probability that the instantaneous power exceeds the power amplifier’s back-off threshold. This probability is referred to as the outage probability, as it indicates the likelihood of signal peaks surpassing the amplifier’s linear operating range, potentially leading to distortion.  

To properly assess system performance, it is important to have two different SNR definitions. The transmit SNR is given by  
$
\text{SNR}_{tx} = \frac{P}{\sigma_0^2}$,
which measures the ratio of total transmit power to noise power. On the other hand, the received SNR is defined as  
$
\text{SNR}_{rx} = \frac{E/T}{\sigma_0^2} = \frac{P\delta}{\sigma_0^2}$,
which accounts for the actual energy per symbol at the receiver. Distinguishing between these two metrics is crucial because system behavior changes significantly with different values of \( \delta \). In standard Nyquist signaling scenarios, \(\delta = 1\), these two SNR definitions are identical, eliminating the need for separate consideration.

In this section, we focus on the PAPR characteristics of FTN signaling with certain time-domain power allocation schemes. 
%Uniform power allocation across both spatial and frequency domains implies that the symbols \( a_k[m] \) are statistically independent across different time indices \( m \) and \( n \) for \( m \neq n \). 
% Furthermore, the real and imaginary components of \( a_k[m] \), denoted as \( a_{r,k}[m] \) and \( a_{i,k}[m] \), are also independently and identically distributed (i.i.d.).  
The PAPR is a key performance metric in power-limited systems, and we define it as  
\begin{equation}
    \text{PAPR} = \frac{|x_k(t)|^2}{P_k}, \label{eqn:defpapr}
\end{equation} 
where $P_k$ is the power level allocated to the $k$th transmit antenna. 
Since \( x_k(t) \) follows a cyclostationary random process with period \( \delta T \), as demonstrated in \cite{zhang2022faster}, it exhibits periodic statistical properties over time. In this paper, we assume that \( N \) is sufficiently large for \( x_k(t) \) to be well approximated by a cyclostationary process, allowing us to focus on its power distribution within a single period \( [0, \delta T) \). The power distribution at each time instant varies due to the time-dependent nature of the pulse-shaping filter coefficients \( p(t - m\delta T) \).  

A crucial tool for evaluating PAPR statistics is the complementary cumulative distribution function (CCDF), which quantifies the probability that the instantaneous power exceeds a given threshold. The CCDF of \( |x_k(t)|^2 \) at time \( t \) is defined as  
\begin{equation}
    \mathcal{C}(\gamma; t) = \text{Pr} \left[ |x_k(t)|^2 \geq \gamma \right]. \label{eqn:defccdf}
\end{equation}  
Since this probability varies with \( t \), analyzing the power behavior over a full period provides a more comprehensive understanding of system performance. We also define the average CCDF, obtained by averaging \eqref{eqn:defccdf} over one period:  
\begin{equation}
    \bar{\mathcal{C}}(\gamma) = \frac{1}{\delta T} \int_0^{\delta T} \mathcal{C}(\gamma; t) dt. \label{eqn:aveccdfdef}
\end{equation}  
This averaged CCDF allows for a more general evaluation of FTN PAPR characteristics by smoothing out the time variations in the power distribution, making it a more practical measure for system design and analysis.

\subsection{Optimal Input Covariance Matrix Derivation }
\label{sec:mimonoeffect}
In MIMO communications, assigning different power levels to each antenna can be advantageous. While spatial precoding at the transmitter introduces correlation between the signals transmitted from different antennas, this correlation does not impact the PAPR. Each antenna operates with its own power amplifier, and PAPR is calculated individually for each antenna based on its mean power level, regardless of the correlation.  Therefore, in the rest of this paper we will analyze the PAPR behavior of a SISO system.

\subsubsection{Moderate acceleration factor ($\frac{1}{1+\beta}\leq\delta\leq1$)}
Let the covariance matrix of $\bm{a}_k$ be $\bm{\Sigma}$. By adjusting the covariance matrix we change the input distribution and also the PAPR. 
For fixed $SNR_{tx}$, the capacity-achieving input distribution for MIMO FTN was derived in (55)-(57) of \cite{zhang2022faster}. %The optimal covariance matrix for MIMO FTN has the structure of $\bm{Z}\otimes\bm{G}^{-1}$, where $\bm{Z}$ is the optimal input covariance matrix for the MIMO channel. Based on this result, we can easily boil down the distribution of the MIMO system to the SISO system by letting the antenna number be 1.  
For SISO, this result is related as follows.
\begin{theorem}{\cite{zhang2022faster}}
    The optimal input distribution for SISO FTN with fixed $SNR_{tx}$ is given as 
    \begin{equation}
        \bm{\Sigma}=P_k\delta T\bm{G}^{-1},~ \frac{1}{1+\beta}\leq\delta\leq1, \label{eqn:bdoptmat}
    \end{equation}
    where the eigenvalues of $\bm{G}^{-1}$ are 
    \begin{equation}
    \lambda_i=\frac{1}{\delta T}\sum_{m=-\infty}^{\infty}G\left(\frac{i/N-m}{\delta T}\right), \frac{1}{1+\beta}\leq\delta\leq1.
\end{equation}
\end{theorem}
\begin{proof}
    According to \cite[Lemma 1]{capregamacftn}, the Toeplitz matrix $\bm{G}$ has discrete Fourier transform (DFT) vectors as its eigenvectors asymptotically, as $N$ goes to infinity. Therefore, as $N$ is large enough, we can perform eigendecomposition on $\bm{G}$ with the DFT matrix $\bm{D}$ as 
$\bm{G}=\bm{D}\bm{\Lambda}\bm{D}^\dagger.$
 According to \cite{gray}, $\bm{G}^{-1}$ is an asymptotically Toeplitz matrix. According to \cite{property}, the eigenvalues of a Toeplitz matrix can be approximated by the samples of its generating function, where we use the notation $\mathcal{G}$ as the operation of obtaining generating function. From \cite[(91)-(94)]{zhang2022faster}, we write 
\begin{align}
    \mathcal{G}(\bm{G})&=\frac{1}{\mathcal{G}(\bm{G}^{-1})}=\frac{1}{\delta T}\sum_{m=-\infty}^{\infty}G\left(\frac{f_n-m}{\delta T}\right), \label{eqn:genfunc}
\end{align}
which is periodical with period 1. When $G(f)$ is the continuous time Fourier transform of $g(t)$,  for $i=-M,\dots,M$, we have 
\begin{equation}
    \lambda_i=\frac{1}{\delta T}\sum_{m=-\infty}^{\infty}G\left(\frac{i/N-m}{\delta T}\right), \frac{1}{1+\beta}\leq\delta\leq1.
\end{equation}
\end{proof}
\subsubsection{Small acceleration factor ($\delta<\frac{1}{1+\beta}$)}
With fixed $SNR_{tx}$, the optimal input frequency spectrum for MIMO FTN is derived in \cite{zhang2024capacitypapranalysismimo}. However, the optimal frequency spectrum does not allow for a PAPR analysis in the time domain. Therefore, in this section, we obtain the optimal covariance matrix in time domain and then calculate PAPR of FTN for optimal power allocation. 
\begin{theorem}
    The optimal input covariance matrix for SISO FTN with $\delta<\frac{1}{1+\beta}$ is given as  
    \begin{equation}
        \bm{\Sigma}=P_k\delta T\bm{D}\bm{\Lambda}^{-1}\bm{D}^\dagger, \label{eqn:sdoptcovmat}
    \end{equation}
    where $\bm{D}$ is DFT matrix with  $\left(\bm{D}\right)_{m,n}=\frac{1}{\sqrt{N}}e^{-j2\pi\frac{mn}{N}}$. The matrix $\bm{\Lambda}^{-1}$ is a diagonal matrix with entries $\Tilde{\lambda}_i, i=-M,\dots,M,$ given as 
\begin{align}
\Tilde{\lambda}_i=\phi(f_n)|_{f_n=\frac{i}{N}}=\begin{cases}
        \frac{  T}{G(\frac{i}{N\delta T})(1+\beta)},& |i|<\frac{Z}{2}  \\
        0,&\frac{Z}{2}\leq|i|\leq M
    \end{cases}, \label{eqn:lamboptisd}
\end{align} 
for $\delta<\frac{1}{1+\beta}$.
\end{theorem}
\begin{proof}
When $\delta<\frac{1}{1+\beta}$, the support of the spectrum $\mathcal{G}(\bm{G})$ in \eqref{eqn:genfunc} in one period is in $[-\frac{\delta(1+\beta)}{2}, \frac{\delta(1+\beta)}{2}]$. Then, some of the samples on the spectrum, or equivalently some of the $\lambda_i$'s will be zero, leading to $\lambda_i^{-1}$ and $\bm{\Sigma}$ in \eqref{eqn:sdoptcovmat} being undefined. This is why we need a new solution method. 

To alleviate the above mentioned problem and to compute $\lambda_i^{-1}$, we use the optimum input frequency spectrum derived in \cite[(25)]{zhang2024capacitypapranalysismimo}. The optimum input spectrum is given as 
\begin{equation}
    \phi(f_n)=\begin{cases}
        \frac{ T}{G(\frac{f_n}{\delta T})(1+\beta)},&~f_n\in\left[-\frac{\delta(1+\beta)}{2}, \frac{\delta(1+\beta)}{2}\right] \\
        0,&~f_n\in \left[-\frac{1}{2}, -\frac{\delta(1+\beta)}{2}\right] \bigcup \left[\frac{\delta(1+\beta)}{2}, \frac{1}{2}\right]
    \end{cases}
\end{equation}
for $\delta<\frac{1}{1+\beta}$.
Then we evenly sample $\phi(f_n)$ in $[-\frac{1}{2}, \frac{1}{2}]$  $N$ times to get the entries of $\bm{\Lambda}^{-1}$. Assume that the number of non-zero samples is $Z$, then   
\begin{align}
\Tilde{\lambda}_i=\phi(f_n)|_{f_n=\frac{i}{N}}=\begin{cases}
        \frac{  T}{G(\frac{i}{N\delta T})(1+\beta)},& |i|<\frac{Z}{2}  \\
        0,&\frac{Z}{2}\leq|i|\leq M
    \end{cases}, 
\end{align}
for $\delta<\frac{1}{1+\beta}$.
\end{proof}

\begin{corollary}
     When $SNR_{rx}$ is fixed, we can obtain the optimal covariance matrix by simply replacing $P_k$ with $\frac{E}{\delta T}$ in \eqref{eqn:bdoptmat} and \eqref{eqn:sdoptcovmat}.
\end{corollary}

\subsection{Average CCDF for FTN with Gaussian Symbols}

In the previous subsection, we have calculated all the $\lambda_i$ values. Using this information, in this subsection, we analyze PAPR for $\delta<\frac{1}{1+\beta}$ and study fixed transmit SNR and fixed received SNR cases individually.  

\subsubsection{Fixed transmit SNR}
Assume that the data symbols $a_k[n]$ are drawn from a complex Gaussian distribution, and  $a_{r,k}[n]$ and $a_{i,k}[n]$ are independent of each other. First of all, we observe that $x_k(t)$ is a cyclostationary Gaussian process, since for any time instant $\tau$, $x_k(\tau)$ is a linear combination of Gaussian symbols. Moreover, $x_{r,k}(t)$ and $x_{i,k}(t)$ are also Gaussian processes. Assume that the covariance matrices for $\bm{a}_{r,k}$ and $\bm{a}_{i,k}$ are $\bm{\Sigma}_{r}$ and $\bm{\Sigma}_{i}$ respectively. 
The variance of $|x_k(\tau)|^2$, which is  $|x_k(t)|^2$ evaluated at time $\tau$ is given by 
\begin{align}
    \mathbb{E}\left[|x_k(\tau)|^2\right]&=\mathbb{E}\left[x^2_{r,k}(\tau)\right]+\mathbb{E}\left[x^2_{i,k}(\tau)\right] \\
    &=\mathbb{E}\left[\left(\sum_{m=-M}^Ma_{r,k}p(\tau-m\delta T)\right)^2\right]\notag\\
    &\quad\quad~+\mathbb{E}\left[\left(\sum_{m=-M}^Ma_{i,k}p(\tau-m\delta T)\right)^2\right]\\
    &=\mathbb{E}\left[\left(\bm{p}_\tau^T\bm{a}_{r,k}\right)^2\right] +\mathbb{E}\left[\left(\bm{p}_\tau^T\bm{a}_{i,k}\right)^2\right]\\
    &=\bm{p}_\tau^T\mathbb{E}[\bm{a}_{r,k}\bm{a}_{r,k}^T]\bm{p}_\tau + \bm{p}_\tau^T\mathbb{E}[\bm{a}_{i,k}\bm{a}_{i,k}^T]\bm{p}_\tau\\
    &=\bm{p}^T_\tau\bm{\Sigma}_{r}\bm{p}_\tau + \bm{p}_\tau^T\bm{\Sigma}_{i}\bm{p}_\tau,
\end{align}
where $\bm{p}_\tau=[p(\tau+M\delta T), p(\tau+(M-1)\delta T),\dots,p(\tau-M\delta T)]$.
We also know that 
\begin{align}
    \mathbb{E}\left[|x_k(\tau)|^2\right]&=\mathbb{E}\left[\left(\sum_{m=-M}^Ma_{k}p(\tau-m\delta T)\right)^2\right]\\
    &=\bm{p}^T_\tau\bm{\Sigma}\bm{p}_\tau.
\end{align}
We can see that the variance of $|x_k(\tau)|^2$ is composed of the variance of $x_{r,k}^2(\tau)$ and $x_{i,k}^2(\tau)$, the variance of $x_{r,k}^2(\tau)$ and $x_{i,k}^2(\tau)$ are the average power allocated to the real and imaginary parts of the complex channel.  In rich-scattering environments, the communication channel can be modeled as the Rayleigh fading channel, which is a complex channel with equal weights in its real and imaginary parts. Therefore, without loss of generality, we can assume that the real part and the imaginary part of the signal $x_k(t)$ have the same distribution, and $\bm{p}^T\bm{\Sigma}_{r}\bm{p}=\bm{p}^T\bm{\Sigma}_{i}\bm{p}$.
As a result, we can see that random variable $|x_k(\tau)|^2$ is also Rayleigh distributed, and the CCDF of the instantaneous power is given as 
\begin{align}
    \bar{\mathcal{C}}(\gamma)=\frac{1}{\delta T}\int_0^{\delta T} \text{exp}\left(-\frac{\gamma}{\bm{p}^T_\tau\bm{\Sigma}\bm{p}_\tau}\right)d\tau. \label{eqn:ccdfgausdiscrete}
\end{align}
Then, we can plug \eqref{eqn:sdoptcovmat} into \eqref{eqn:ccdfgausdiscrete} and get 
\begin{align}
    \bar{\mathcal{C}}(\gamma)=\frac{1}{\delta T}\int_0^{\delta T} \text{exp}\left(-\frac{\gamma}{P_k\delta T\bm{p}^T_\tau\bm{D}\bm{\Lambda}\bm{D}^\dagger\bm{p}_\tau}\right)d\tau. \label{eqn:origccdf}
\end{align}
Notice that the multiplication $\bm{p}_\tau^T\bm{D}$ is equivalent to performing DFT on the vector $\bm{p}_\tau$. We call this product $\bm{q}_\tau$. Since we have sufficiently large $N$,  the transform becomes  
\begin{align}
    (\bm{q}_\tau)_n&=\frac{1}{\sqrt{N}}\sum_{m=-M}^Mp(\tau-m\delta T)e^{-j2\pi\frac{mn}{N}} \\
    &=\frac{1}{\sqrt{N}\delta T}\sum_{v=-\infty}^\infty\sqrt{G\left(\frac{n/N-v}{\delta T}\right)}e^{j2\pi\frac{n/N-v}{\delta T}\tau}. \label{eqn:29}
\end{align}
In \eqref{eqn:29}, we observe that the term 
\begin{equation}
    \frac{1}{\delta T}\sum_{v=-\infty}^\infty\sqrt{G\left(\frac{f_n-v}{\delta T}\right)}e^{j2\pi\frac{f_n-v}{\delta T}\tau}, f_n\in\left[-\frac{1}{2}, \frac{1}{2}\right],\notag
\end{equation} is the sum of shifted and scaled versions of $G(f)$. The vector $\bm{q}_\tau$ is obtained by evenly sampling the spectrum in one period $N$ times.  Since $\delta<\frac{1}{1+\beta}$, we then have
\begin{align}
    \left(\bm{q}_\tau\right)_n&=\frac{1}{\sqrt{N}\delta T}\sqrt{G\left(\frac{n}{N\delta T}\right)}e^{j2\pi\frac{n}{N\delta T}\tau}.\label{eqn:qtau}
\end{align}
In the meantime, using the relationship $\bm{p}^T_\tau\bm{D}=\bm{q}_\tau$, and plugging \eqref{eqn:qtau} and \eqref{eqn:lamboptisd} into \eqref{eqn:origccdf} to get
\begin{align}
    \lefteqn{\bar{\mathcal{C}}(\gamma)=\frac{1}{\delta T}\int_0^{\delta T} \text{exp}\left(-\frac{\gamma}{P_k\delta T\sum_{n=-M}^{M}|\left(\bm{q}_\tau\right)_n|^2\lambda_n}\right)d\tau} \\
    &=\frac{1}{\delta T}\int_0^{\delta T} \text{exp}\left(-\frac{\gamma}{\frac{P_k\delta T}{N(\delta T)^2}\sum_{n=-M}^{M}G\left(\frac{n}{N\delta T}\right)\frac{(1+\beta)T}{G(\frac{n}{N\delta T})}}\right)d\tau. \label{eqn:unfinished}
    \end{align}
    Note that as $\delta<\frac{1}{1+\beta}$, there will be some samples $G(\frac{n}{N\delta T})$ which are zero. Denoting the number of non-zero samples with $Z$,  \eqref{eqn:unfinished} becomes 
    \begin{align}
    \bar{\mathcal{C}}(\gamma)&=\frac{1}{\delta T}\int_0^{\delta T} \text{exp}\left(-\frac{\gamma}{\frac{P_k}{\delta(1+\beta)}\frac{Z}{N}}\right)dt\\
    &\overset{(a)}{\approx}\frac{1}{\delta T}\int_0^{\delta T} \text{exp}\left(-\frac{\gamma}{P_k}\right)dt\\
    &=\text{exp}\left(-\frac{\gamma}{P_k}\right).
\end{align}
As $N$ is sufficiently large, the ratio between the number of non-zero samples $Z$ and the number of total samples $N$ can approximate the ratio between the length of the support of $G(\frac{f_n}{\delta T})$, which is $\delta(1+\beta)$, and its period 1. Therefore, in (a), the ratio $\frac{Z}{N}$ can be approximated as $\delta(1+\beta)$. As we can see, when $\delta<\frac{1}{1+\beta}$, the average CCDF for instantaneous power is irrelevant to the value of $\delta$. According to \cite[Remark 4]{zhang2024capacitypapranalysismimo}, the average CCDF of instantaneous power has the same behavior as the average CCDF of PAPR. So we have the following theorem.
\begin{theorem}
\label{thm:thm1}
    If $SNR_{tx}$ is fixed and Gaussian symbols are used, under the optimal power allocation scheme, as $\delta$ approaches zero, the average CCDF of PAPR, $\bar{\mathcal{C}}(\gamma)$,  does not change with $\delta$ and is equal to  
    \begin{equation}
        \bar{\mathcal{C}}(\gamma)=\text{exp}\left(-\frac{\gamma}{P_k}\right), \delta<\frac{1}{\beta}.\label{eqn:thm1eq}
    \end{equation}
\end{theorem}
% Based on the definition of \eqref{eqn:defpapr}, we plug in the definition of $x_k(t)$ from \ref{eqn:defxt} to $|x_k(t)|^2$ and get 
% \begin{align}
    
% \end{align}

\subsubsection{Fixed received SNR}

When \( SNR_{rx} \) remains constant, we can substitute \( P_k \delta T \) with \( E \). This alteration causes the average CCDF of instantaneous power to exhibit distinct characteristics compared to the fixed \( SNR_{tx} \) case as stated next. 

\begin{corollary}
If $SNR_{rx}$ is fixed and Gaussian symbols are used, under the optimal power allocation scheme, the average CCDF of PAPR, $\bar{\mathcal{C}}(\gamma)$,  becomes
\begin{equation}
        \bar{\mathcal{C}}(\gamma)=\text{exp}\left(-\frac{\gamma\delta T}{E}\right), \delta<\frac{1}{\beta},\label{eqn:rxsnr}
    \end{equation}
and as $\delta$
approaches zero, $\bar{\mathcal{C}}(\gamma)$ converges to 1.
\end{corollary} 
\begin{proof}
    We plug in the relation $P_k=\frac{E}{\delta T}$ into \eqref{eqn:thm1eq} and get \eqref{eqn:rxsnr}. From \eqref{eqn:rxsnr} we can see easily that as $\delta$ approaches 0, the average CCDF $\bar{\mathcal{C}}(\gamma)$ approaches 1.
\end{proof}
In summary, for FTN signaling using Gaussian symbols with a fixed \( SNR_{rx} \), the average CCDF curve for instantaneous power becomes a horizontal line as \( \delta \to 0 \). Similarly, according to \cite[Remark 4]{zhang2024capacitypapranalysismimo}, the behavior of average CCDF for PAPR is the same as the behavior of average CCDF for instantaneous power. We conclude that the average CCDF curve for PAPR also approaches a horizontal line as \( \delta \to 0 \).

\section{Simulation Results}

To evaluate the feasibility of FTN signaling, we perform simulations on the PAPR performance for various acceleration rates. Additionally, we present both the empirical and theoretical average CCDF of the PAPR given in \eqref{eqn:defccdf} for optimal power allocation for small acceleration factors. 
%The definition of the average CCDF of instantaneous power is provided in \eqref{eqn:aveccdfdef}. By substituting \( \gamma \) with \( \gamma' P_k \), we derive the average CCDF of PAPR, as explained in \cite[Remark 4]{zhang2024capacitypapranalysismimo}.
In addition, we compare these results with three other power allocation schemes suggested in \cite{zhang2024capacitypapranalysismimo}. These schemes are time inverse power allocation, uniform power allocation in frequency, and uniform power allocation. 
Time inverse power allocation means that $\bm{G}^{-1}$ precoding is performed against ISI due to FTN, but uniform power allocation is assumed among transmit antennas. Uniform power allocation in frequency assumes that the covariance matrix  $\bm{\Sigma}$ is identity matrix, but optimal power allocation, i.e., waterfilling, is used among transmit antennas.
Finally, uniform power allocation implies that the covariance matrix $\bm{\Sigma}$ is identity matrix, and there is uniform power allocation among transmit antennas.
We also compare our MIMO simulations with the theoretical analysis we find in \eqref{eqn:thm1eq} and \eqref{eqn:rxsnr} for SISO.

\begin{figure}[t]
    \centering
    \includegraphics[scale=0.57]{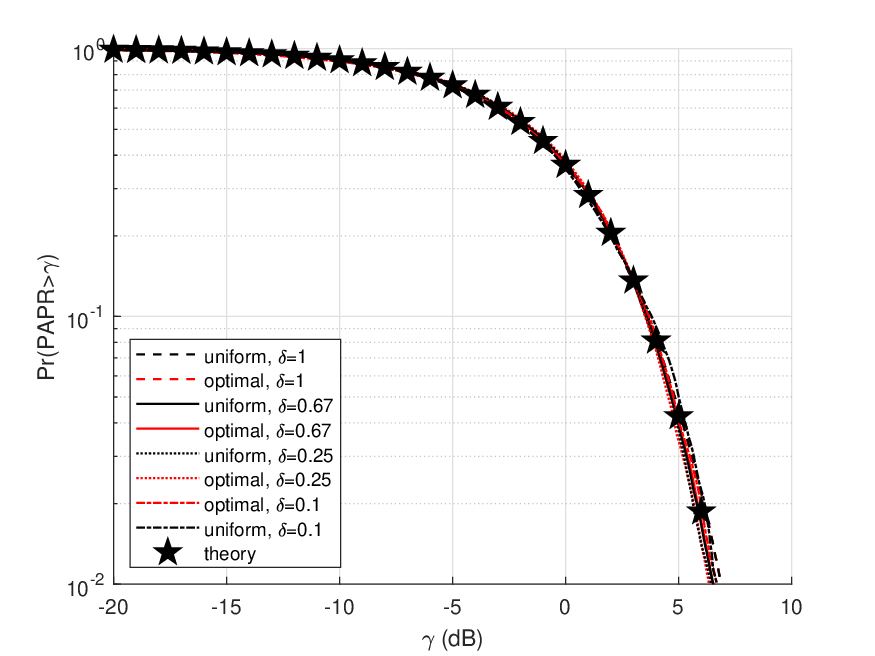}
    \caption{CCDF of average PAPR of SISO FTN signaling with uniform and optimal power allocation schemes for different $\delta$ values. Transmitted SNR is fixed and symbols are drawn from a Gaussian distribution with zero mean and unit variance.}
    \label{fig:ccdfsisogaus}
\end{figure}

In the simulations we plot \eqref{eqn:defccdf} when the symbol period is set to \( T = 0.01 \), and the results are averaged over 1000 random channel realizations. The MIMO channel coefficients \( h_{lk} \) are modeled as independent, complex Gaussian random variables with a distribution of \( \mathcal{CN}\left(0, \frac{1}{K}\right) \). We transmit 2000 symbols using square root-raised cosine pulses with roll-off factor $\beta$ for pulse shaping. 
\begin{figure}[t]
    \centering
    \includegraphics[scale=0.57]{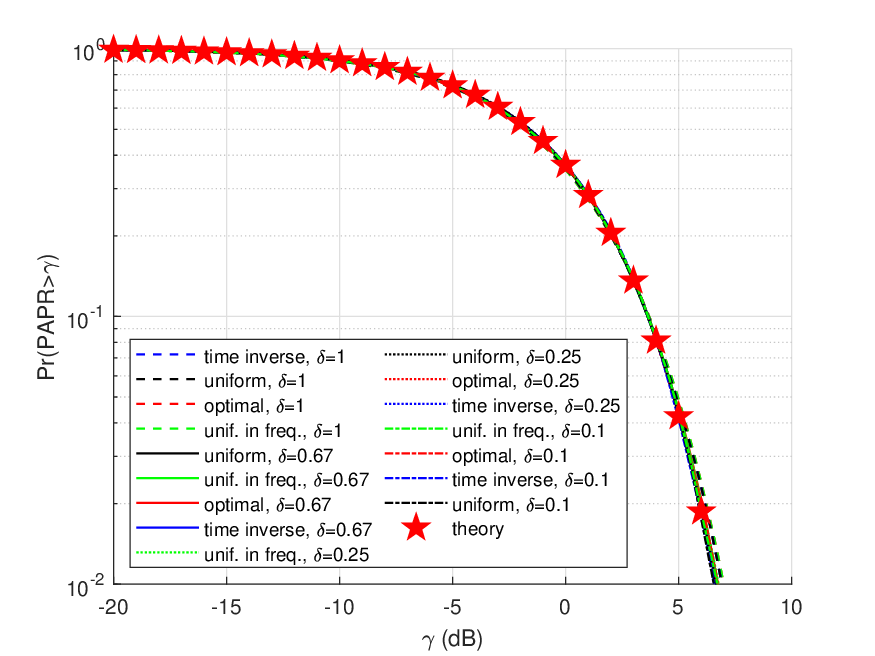}
    \caption{CCDF of PAPR of 20$\times$20 MIMO FTN signaling with different power allocation schemes and $\delta$ values. Transmitted SNR is fixed and symbols are drawn from Gaussian distribution with zero mean and unit variance. }
    \label{fig:ccdf20mimogaus}
\end{figure}

In Fig. \ref{fig:ccdfsisogaus}, we obtain the CCDF of SISO FTN using Gaussian symbols with fixed $SNR_{tx}$, which means the symbols before precoding are generated with distribution $\mathcal{CN}(0, P\delta T)$.  The transmitted signal $x(t)$ in SISO FTN also has the form of \eqref{eqn:defxt}, namely, $x(t)=\sum_{n=0}^{N-1}a[n]p(t-n\delta T)$. From the figure we can see that the value of $\delta$ has no effect on the PAPR behavior of SISO FTN as depicted in Theorem~\ref{thm:thm1}.    After we apply precoding to the input symbols, the resulting process has the same power levels. Moreover, the resulting process is still a Gaussian process since we use linear precoding. Therefore, optimal power allocation overlaps with uniform power allocation scheme.  Since the symbols are Gaussian distributed, all the curves are packed together as expected. 

In Fig. \ref{fig:ccdf20mimogaus}, we investigate the PAPR performance of MIMO FTN with Gaussian symbols for fixed $SNR_{tx}$. We can see that for MIMO FTN as well, $\delta$ has no influence on the PAPR performance.    By comparing Fig.~\ref{fig:ccdf20mimogaus} with Fig.~\ref{fig:ccdfsisogaus} we can see that the results for MIMO overlap with the SISO results. Therefore, the number of antennas has no effect on the PAPR performance as discussed in Section \ref{sec:mimonoeffect}.

\begin{figure}[t]
    \centering
    \includegraphics[scale=0.57]{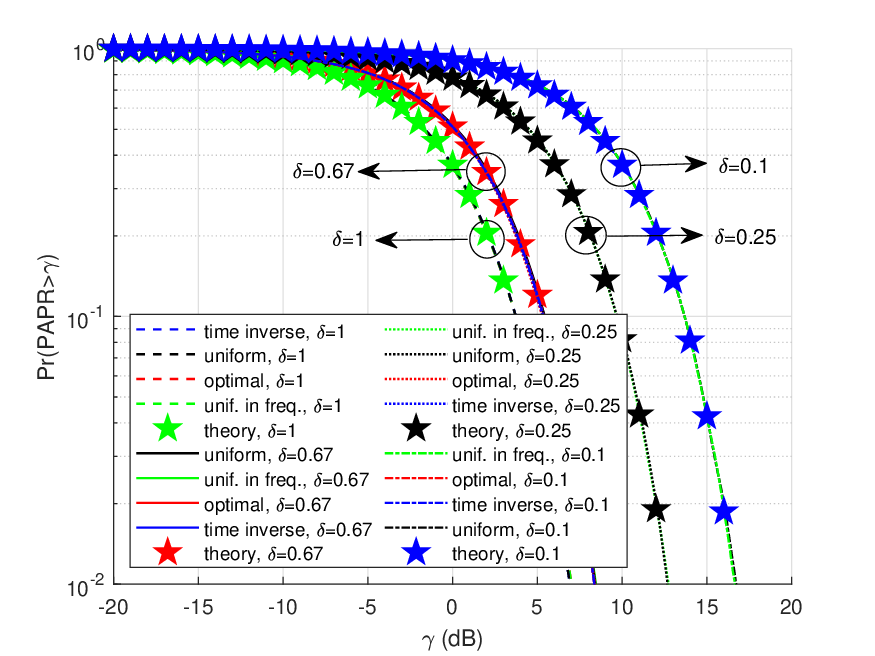}
    \caption{CCDF of PAPR of 20$\times$20 MIMO FTN signaling with different power allocation schemes and $\delta$ values. Received SNR is fixed and symbols are drawn from Gaussian distribution with zero mean and unit variance.}
    \label{fig:ccdf20mimogausrx}
\end{figure}

In Fig.~\ref{fig:ccdf20mimogausrx}, we study the PAPR performance of MIMO FTN with fixed $SNR_{rx}$. First, we observe that MIMO simulations coincide with the theory computed for SISO. Secondly  we  see that the PAPR performance gets worse as $\delta$ decreases. Keeping the symbol energy constant increases the  power level at the transmitter, leading to a severe PAPR increase. On the other hand, we notice that for the same $\delta$ value, the average CCDF of PAPR for all kinds of power allocations overlap with each other. This aligns with what we observed in the fixed $SNR_{tx}$ cases. After adjusting the input distribution, the process is still Gaussian. Therefore for the same $\delta$, power allocation scheme does not change the PAPR distribution. 

% \vspace{-0.3cm}
\section{Conclusion}
In this paper, we study the PAPR behavior of MIMO FTN signaling using Gaussian symbols for fixed transmit SNR and for fixed received SNR. The optimal input distribution is used and we observe that the PAPR behavior for optimal power allocation resembles that of uniform power allocation for fixed $SNR_{tx}$. On the other hand, we find that for fixed $SNR_{rx}$, the PAPR performance becomes unbounded. Moreover, the PAPR behavior only depends on the acceleration rate and the power constraint but not on the power allocation scheme.
For future works, we intend to investigate the PAPR behavior of MIMO FTN with practical constellations such as QPSK using non-uniform power allocation schemes.

% \section*{Acknowledgment}

% \begin{thebibliography}{9}

% \bibitem{Laport:LaTeX}
% L.~Lamport,
%   \emph{\LaTeX: A Document Preparation System,} 
%   Addison-Wesley, Reading, Massachusetts, USA, 2nd~ed., 1994. 

% \bibitem{GMS:LaTeXComp}
% F.~Mittelbach, M,~Goossens, J.~Braams, D.~Carlisle, and
% C.~Rowley, \emph{The {\LaTeX} Companion,} Addison-Wesley,
% Reading, Massachusetts, USA, 2nd~ed., 2004.

% \bibitem{oetiker_latex}
% T.~Oetiker, H.~Partl, I.~Hyna, and E.~Schlegl, \emph{The Not So Short
%   Introduction to {\LaTeX2e}}, version 5.06, Jun.~20, 2016. [Online].
%   Available: \url{https://tobi.oetiker.ch/lshort/}

% \bibitem{typesetmoser}
% S.~M. Moser, \emph{How to Typeset Equations in {\LaTeX}}, version 4.6,
%   Sep. 29, 2017. [Online]. Available:
%   \url{http://moser-isi.ethz.ch/manuals.html#eqlatex}

% \bibitem{IEEE:pdfsettings}
% IEEE, \emph{Preparing Conference Content for the IEEE Xplore Digital
%   Library.} [Online]. Available:
%   \url{http://www.ieee.org/conferences_events/conferences/organizers/pubs/preparing_content.html}

% \bibitem{IEEE:AuthorToolbox}
% IEEE, \emph{Author Digital Toolbox.} [Online.] Available:
%   \url{http://www.ieee.org/publications_standards/publications/authors/authors_journals.html}

% \end{thebibliography}

\bibliographystyle{IEEEtran}

\bibliography{main}

\end{document}